\documentclass[11pt, fleqn]{article}
\usepackage[english]{babel}

\usepackage[T1]{fontenc} 

\usepackage[lmargin=1.1in,rmargin=1.1in,bottom=1.3in,top=1.3in,
twoside=False]{geometry}

\usepackage{mathrsfs}
\usepackage{relsize,xspace}
 \usepackage{xcolor}
 \usepackage{mathtools}
 \usepackage{todonotes}
 \usepackage{comment}
\usepackage{microtype}
\usepackage{amsmath}
\usepackage{amssymb}
\usepackage{amsfonts}
\usepackage{stmaryrd}
\usepackage{bm}
\usepackage{tikz}
\usepackage{refcount}
\usepackage{wrapfig}

\usepackage{marginnote}

\definecolor{sz}{rgb}{0.1,0.2,0.6}
\definecolor{blue}{rgb}{0.1,0.2,0.7}
\definecolor{brown}{rgb}{0.6,0.6,0.2}

\usepackage[amsmath,thmmarks,hyperref]{ntheorem}
\usepackage{cleveref}
\usepackage{algorithm}
\usepackage{algorithmic}
\usepackage{subcaption}
\usepackage{xcolor}
\usepackage{pgffor}
\usepackage{balance}
\usepackage{multirow}
\usepackage[absolute]{textpos}

\crefformat{page}{#2page~#1#3}%
\Crefformat{page}{#2Page~#1#3}%
\crefformat{equation}{#2(#1)#3}%
\Crefformat{equation}{#2(#1)#3}%
\crefformat{figure}{#2Figure~#1#3}%
\Crefformat{figure}{#2Figure~#1#3}%
\crefformat{section}{#2Section~#1#3}
\Crefformat{section}{#2Section~#1#3}
\crefformat{appendix}{#2Appendix~#1#3}
\Crefformat{appendix}{#2Appendix~#1#3}
\crefformat{chapter}{#2Chapter~#1#3}
\Crefformat{chapter}{#2Chapter~#1#3}
\crefformat{chapter*}{#2Chapter~#1#3}
\Crefformat{chapter*}{#2Chapter~#1#3}
\crefformat{part}{#2Part~#1#3}
\Crefformat{part}{#2Part~#1#3}
\crefformat{enumi}{#2(#1)#3}
\Crefformat{enumi}{#2(#1)#3}

\usepackage{enumerate}

\usepackage{latexsym}

\theoremnumbering{arabic}
\theoremstyle{plain}
\theoremsymbol{}
\theorembodyfont{\itshape}
\theoremheaderfont{\normalfont\bfseries}
\theoremseparator{.}

\newtheorem{theorem}{Theorem}
\crefformat{theorem}{#2Theorem~#1#3}
\Crefformat{theorem}{#2Theorem~#1#3}
\newcommand{\newtheoremwithcrefformat}[2]{%
  \newtheorem{#1}[lemma]{#2}%
  \crefformat{#1}{##2\MakeUppercase#1~##1##3}%
  \Crefformat{#1}{##2\MakeUppercase#1~##1##3}%
}
\newcommand{\newseptheoremwithcrefformat}[2]{%
  \newtheorem{#1}{#2}%
  \crefformat{#1}{##2\MakeUppercase#1~##1##3}%
  \Crefformat{#1}{##2\MakeUppercase#1~##1##3}%
}

\newseptheoremwithcrefformat{lemma}{Lemma}
\newtheoremwithcrefformat{proposition}{Proposition}
\newtheoremwithcrefformat{observation}{Observation}
\newtheoremwithcrefformat{conjecture}{Conjecture}
\newtheoremwithcrefformat{corollary}{Corollary}
\newseptheoremwithcrefformat{claim}{Claim}
\newseptheoremwithcrefformat{goal}{Goal}
\theorembodyfont{\upshape}
\newtheoremwithcrefformat{example}{Example}
\newtheoremwithcrefformat{remark}{Remark}
\newseptheoremwithcrefformat{definition}{Definition}
\newseptheoremwithcrefformat{assumption}{Assumption}

\theoremstyle{nonumberplain}
\theoremheaderfont{\scshape}
\theorembodyfont{\normalfont}
\theoremsymbol{\ensuremath{\square}}

\newtheorem{proof}{Proof}

\theoremsymbol{\ensuremath{\lrcorner}}

\newcommand{\Oof}{\mathcal{O}}
\newcommand{\CCC}{\mathscr{C}}
\newcommand{\depthone}{\mbox{depth-1}~}
\newcommand{\minor}{\preceq}
\newcommand{\N}{\mathbb{N}}
\newcommand{\R}{\mathbb{R}}

\hypersetup{colorlinks=true,linkcolor=blue,citecolor=blue, urlcolor=blue}

\title{Distributed Dominating Set Approximations\\ beyond Planar Graphs\thanks{Preliminary versions of this paper
appeared as \cite{amirilog} and \cite{Amiri2016}. 
The research of Saeed Akhoondian Amiri 
was partially supported by the European Research
Council (ERC) under the European Union's Horizon 2020 research and innovation programme (ERC consolidator grant DISTRUCT,
agreement No 648527). 
Sebastian Siebertz was partially supported by the European Research
Council (ERC) under the European Union's Horizon 2020 research and innovation programme (ERC consolidator grant DISTRUCT,
agreement No 648527) and by the National Science Centre of Poland via POLONEZ grant agreement UMO-2015/19/P/ST6/03998, 
which has received funding from the European Union's Horizon 2020 research and 
innovation programme (Marie Sk\l odowska-Curie grant agreement No.\ 665778).}}

\author{Saeed Akhoondian Amiri\\
Max-Planck-Institut f\"ur Informatik, Germany\\
\href{mailto:samiri@mpi-inf.mpg.de}{samiri@mpi-inf.mpg.de}
\and 
Stefan Schmid\\
Universit\"at Wien, Austria\\
\href{mailto:stefan_schmid@univie.ac}{stefan\_schmid@univie.ac}
\and
Sebastian Siebertz\\
Humboldt Universit\"at zu Berlin, Germany\\
\href{mailto:sebastian.siebertz@hu-berlin.de}{sebastian.siebertz@hu-berlin.de}}

\begin{document}

\maketitle

\begin{abstract}
\noindent The Minimum Dominating Set (MDS) problem is a fundamental
and challenging problem in distributed computing. While it is well-known that  
minimum dominating sets cannot be well approximated locally on general graphs, 
over the last years, there has been much progress on
computing good local approximations on sparse graphs, 
and in particular on planar graphs. 
In this paper we study distributed and deterministic MDS
approximation algorithms for graph classes beyond planar 
graphs. 
In particular, we show that existing approximation bounds
for planar graphs can be lifted to bounded genus graphs 
and more general graphs, which we call locally embeddable graphs, 
and present 
\begin{enumerate}
\item a local constant-time, constant-factor 
MDS approximation algorithm on locally embeddable graphs, and 
\item  a local $\Oof(\log^*{n})$-time (1+$\epsilon$)-approximation 
scheme for any $\epsilon>0$ on graphs of bounded genus.
\end{enumerate}  
Our main technical contribution is a new analysis of
a slightly modified variant of an existing algorithm by Lenzen et al.
Interestingly, unlike existing proofs for planar graphs, our
analysis does not rely on direct topological arguments, but on
combinatorial density arguments only.
\end{abstract}

\begin{picture}(0,0) \put(395,-98)
{\hbox{\includegraphics[scale=0.25]{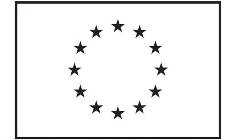}}} \end{picture} 
\vspace{-0.8cm}

\section{Introduction}\label{sec:intro}

This paper attends to the Minimum Dominating
Set (MDS) problem, an intensively 
studied graph theoretic problem in computer science in general,
as well as in distributed computing.

A dominating set~$D$ in a graph~$G$ is a set of vertices
such that every vertex of~$G$ either lies in~$D$ or is adjacent to a vertex in
$D$. 
Finding a minimum dominating set is 
NP-complete~\cite{karp1972reducibility}, even on
planar graphs of maximum degree~$3$ (cf.\ [GT2] in~\cite{michael1979computers}). Consequently, attention 
has shifted from computing exact solutions to approximating
near optimal dominating sets. The simple greedy algorithm on 
$n$-vertex graphs 
 computes an $\ln n$ approximation of a minimum
dominating set~\cite{johnson1974approximation,lovasz1975ratio}, 
and for general graphs this algorithm is near optimal --
it is NP-hard to approximate minimum dominating sets within factor 
$(1-\epsilon)\cdot \ln n$ for every~$\epsilon>0$~\cite{dinur2014analytical}.

The approach of algorithmic graph structure theory is to exploit
structural properties of restricted graph classes for the design 
of efficient algorithms. For the dominating set problem this has
led to a PTAS on planar graphs~\cite{Baker:1994:AAN:174644.174650}, 
minor closed classes of graphs 
with locally bounded tree-width~\cite{eppstein2000diameter}, 
graphs with excluded minors~\cite{grohe2003local}, and 
most generally, on every graph class with subexponential expansion~\cite{har2015approximation}. The problem admits a constant factor 
approximation on classes of bounded 
arboricity~\cite{bansal2017tight} and an $\Oof(\ln k)$ 
approximation (where $k$ denotes the size of a minimum dominating
set) on classes of bounded VC-dimension~\cite{bronnimann1995almost,even2005hitting}. On the other hand, it
is unlikely that polynomial-time constant factor approximations 
exist even on $K_{3,3}$-free 
graphs~\cite{siebertz2019greedy}. The general goal of algorithmic
graph structure theory is to identify the broadest graph classes 
on which certain algorithmic techniques can be applied and hence
lead to efficient algorithms for problems that are hard on general
graphs. These limits of tractability are often captured by 
abstract notions, such as expansion, arboricity or VC-dimension 
of graph classes.

In this paper, we study the \emph{distributed} time complexity of
finding dominating sets, in the classic \textit{LOCAL model} 
of distributed computing~\cite{Linial:1992:LDG:130563.130578}. 
It is known that finding small dominating sets locally is hard: 
Kuhn et al.~\cite{Kuhn:2016:LCL:2906142.2742012} show that in~$r$ rounds the 
MDS problem on an~$n$-vertex graphs of maximum degree 
$\Delta$  can only be approximated within factor~$\Omega(n^{c/r^2}/r)$
and~$\Omega(\Delta^{1/(r+1)}/r)$, where~$c$ is a constant. 
This implies that, in general, to achieve a constant approximation ratio, 
every distributed algorithm requires at least~$\Omega(\sqrt{\log
    n/\log \log n})$ and~$\Omega(\log \Delta/\log \log \Delta)$ communication rounds. 
The currently best results for general graphs are by
Kuhn et al.~\cite{Kuhn:2016:LCL:2906142.2742012} who present a~$(1+\epsilon)\ln \Delta$-approximation in~$\Oof(\log(n)/\epsilon)$ rounds for any~$\epsilon>0$,
and by Barenboim et al.~\cite{barenboim2014fast}
who present a deterministic $\Oof((\log n)^{k-1})$-time algorithm that provides an
$\Oof(n^{1/k})$-approximation, for any integer parameter $k \ge 2$.

For sparse graphs, the situation is more promising (an inclusion
diagram of the graph classes mentioned in the following paragraph is
depicted in Figure~\ref{fig:classes}, for formal definitions we refer to the referenced papers). 
For graphs of arboricity~$a$, Lenzen and Wattenhofer~\cite{ds-arbor}
present a forest decomposition algorithm achieving a factor~$\Oof(a^2)$ approximation
in randomized time~$\Oof(\log n)$, and a deterministic~$\Oof(a \log
\Delta)$ approximation algorithm
requiring $\Oof(\log \Delta)$ rounds. Graphs of bounded arboricity include all graphs which exclude a fixed graph as a (topological) minor and in particular, all planar graphs and any class of bounded genus. 
Amiri et al.~\cite{amiri2017distributed} provide a deterministic 
$\Oof(\log n)$ time constant factor approximation algorithm on 
classes of bounded expansion (which extends also to connected
dominating sets). The notion of bounded expansion 
offers an abstract definition of uniform sparseness in graphs, which
is based on bounding the density of shallow minors 
(these notions will be defined formally
in the next section). 
Czygrinow et al.~\cite{fast-planar} show
that for any given~$\epsilon>0$,~$(1+\epsilon)$-approximations of a maximum independent
set, a maximum matching, and a minimum dominating set, can be computed in
$\Oof(\log^* n)$ rounds in planar graphs, which is asymptotically optimal~\cite{ds-alternative-lowerbound}.
Lenzen et al.~\cite{ds-planar} proposed a constant factor
approximation on planar graphs that can be computed locally in a
constant number of communication rounds. A finer analysis of
Wawrzyniak~\cite{better-upper-planar} showed that the algorithm of
Lenzen et al.\ in fact computes a~$52$-approximation of a minimum
dominating set. 
Wawrzyniak~\cite{wawrzyniak2013brief} also showed
that message sizes of $\mathcal{O}(\log n)$ suffice to give a
constant factor approximation on planar graphs in a constant number
of rounds. 
In terms of lower bounds, Hilke et al.~\cite{ds-ba} show that there is no 
deterministic local algorithm (constant-time distributed graph algorithm) that 
finds a~$(7-\epsilon)$-approximation of a minimum dominating set on 
planar graphs, for any positive constant~$\epsilon$.

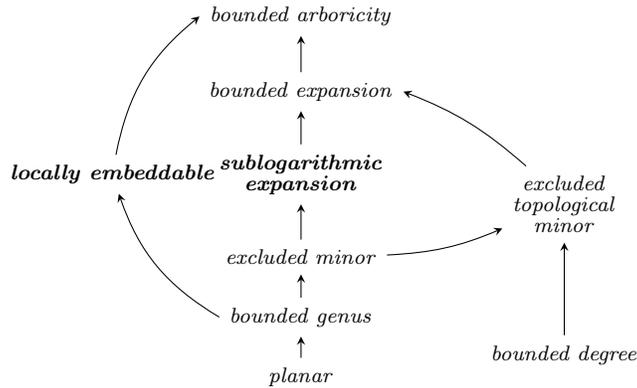
\begin{figure}[ht]
\begin{center}
\begin{tikzpicture}

\node (bd-deg) at (10,-4) {\scriptsize\textit{bounded degree}};
\node[align=center] (topminor) at (10,-2) {\scriptsize\textit{excluded}\\[-2mm]\scriptsize\textit{topological}\\[-2mm] \scriptsize\textit{minor}};
\node (bd-exp) at (6.5,-0.5) {\scriptsize\textit{bounded expansion}};
\node (degenerate) at (6.5,0.5) {\scriptsize\textit{bounded arboricity}};

\node (planar) at (6.5,-4.3) {\scriptsize\textit{planar}};
\node (genus) at (6.5,-3.5) {\scriptsize\textit{bounded genus}};
\node (emb) at (4,-1.6) {\scriptsize\textbf{\textit{locally embeddable}}};
\node (minor) at (6.5,-2.7) {\scriptsize\textit{excluded minor}};
\node[align=center] (subexp) at (6.5,-1.6) {\scriptsize\textbf{\textit{sublogarithmic}}\\[-2mm] \scriptsize\textbf{\textit{expansion}}};


\draw[->,>=stealth] (planar) to (genus);
\draw[->,>=stealth] (genus.west) to[bend left=20] (emb);
\draw[->,>=stealth] (genus) to (minor);
\draw[->,>=stealth] (emb) to[bend left=20] (degenerate.west);
\draw[->,>=stealth] (minor) to (subexp);
\draw[->,>=stealth] (subexp) to (bd-exp);
\draw[->,>=stealth] (bd-deg) to (topminor);
\draw[->,>=stealth] (minor) to[bend right=10] (topminor);
\draw[->,>=stealth] (topminor) to[bend right=10] (bd-exp.east);
\draw[->,>=stealth] (bd-exp) to (degenerate);

\end{tikzpicture}
\end{center}
\caption{Inclusion diagram of sparse graph classes. }
\end{figure}\label{fig:classes}

\subsection{Our Contributions}
 
The first and main contribution of this paper is 
a deterministic and local constant factor approximation 
for MDS on graphs that we call \emph{locally embeddable graphs}. 
A locally embeddable graph~$G$ excludes the complete bipartite graph 
$K_{3,t}$, for some $t\geq 3$, as a \depthone minor, that is, as a minor 
obtained by star contractions, and furthermore satisfies that 
all \depthone minors of $G$ have constant edge density. The most
prominent locally embeddable graph classes are classes of bounded
genus. Concretely, our result implies that MDS can be~$\Oof(g)$-approximated
locally and deterministically on graphs of (both orientable
or non-orientable) genus~$g$. However, also graph classes whose
members do not embed into any fixed surface or which do not
even have bounded expansion can be locally embeddable, 
e.g.\ the class of all $3$-subdivided cliques is locally embeddable 
and this class does not have bounded expansion. Yet,
every locally embeddable class is of bounded degeneracy. 
Apart from generalizing the earlier result of Lenzen et 
al.~\cite{ds-planar} for planar graphs to a larger graph family, 
we introduce new techniques by arguing about densities and
combinatorial properties of shallow minors only and show that 
all topological arguments used in~\cite{ds-planar} can be 
avoided. The abstract notion of local embeddability yields 
exactly the ingredients for these arguments to work and 
therefore offers valuable insights on the limits of algorithmic
techniques. This is a contribution going beyond the mere 
presentation of efficient algorithms for concrete example graph 
classes.

Our second main contribution is the presentation of a 
local and deterministic MDS approximation algorithm with the following
properties. Given a graph $G$ from a fixed class $\CCC$ of graphs 
with sub-logarithmic
expansion, a constant factor approximation of 
an MDS $D$ and any $\epsilon>0$, the algorithm
uses $\Oof(\log^* n)$ rounds and computes from $D$ a 
$(1+\epsilon)$-approximate MDS of $G$ (here, the $\Oof$-notation
hides constants depending on $\epsilon$). Graphs of sub-logarithmic
expansion include all proper minor closed classes, and in particular all 
classes of bounded genus. Our 
methods are based on earlier work of Czygrinow et al.~\cite{fast-planar}. 
In combination with our constant-factor approximation on 
graphs of bounded genus, we obtain $(1+\epsilon)$-approximations
in $\Oof(\log^*n)$ communication rounds on graphs
of bounded genus. In combination with 
Amiri et al.'s result~\cite{amiri2017distributed} on 
graphs of bounded expansion,  
we obtain $(1+\epsilon)$-approximations in $\Oof(\log n)$ deterministic
rounds on graphs of sub-logarithmic expansion. Again, the 
abstract notion of sub-logarithmic expansion constitutes the 
border of applicability of these algorithmic techniques. 

We observe that the methods of Czygrinow et al.~\cite{fast-planar}
for maximum weighted independent set and maximum matching
extend to graphs of sub-logarithmic expansion, however, we focus on
the dominating set problem for the sake of a consistent presentation. 

\subsection{Novelty}

Our main technical contribution is a new analysis of
a sligthly modified variant of the elegant 
algorithm by Lenzen et al.~\cite{ds-planar} for planar graphs. 
As we will show, with a slight modification, the algorithm also
works on locally embeddable graphs, however, the analysis
needs to be changed significantly. 
Prior works by 
Lenzen et al.~\cite{ds-planar} and Wawrzyniak~\cite{better-upper-planar} 
heavily depend 
on topological properties of planar graphs. For example, their analyses
exploit the fact that each cycle in a planar graph
defines an ``inside'' and an ``outside'' region,
without any edges connecting the two; this facilitates a simplified
accounting and comparison to the optimal solution. 
In the case of locally embeddable graphs,
such global, topological properties do not exist.
In contrast, in this paper we leverage the inherent
local properties of our low-density graphs,
which opens a new door to approach the problem. 

A second interesting technique developed in this paper
is based on \emph{preprocessing}: we show that the constants
involved in the approximation can be further improved 
by a local preprocessing step. 


Another feature of our modified algorithm is that
it is \emph{first-order definable}. More precisely, there is a first
order formula $\varphi(x)$ with one free variable, such that in every
planar graph~$G$ the set $D=\{v \in V(G) : G\models\varphi(v)\}$
corresponds exactly to the computed dominating set. 
In particular, the
algorithm can be modified such that 
it does not rely on any \emph{maximum} operations, such as
finding the neighbor of maximal degree. 

\subsection{Organization}
%
\noindent The remainder of this paper is organized as follows. We introduce some preliminaries in \cref{sec:model}.
The constant-factor constant-time approximation result
is presented in \cref{sec:local-approx},
and the $\Oof(\log^* n)$-time approximation scheme is presented in
\cref{sec:star-approx}.
We conclude in \cref{sec:FO}.

\section{Preliminaries}\label{sec:model}

\vspace{2mm}
\noindent \textbf{Graphs.}
We consider finite, undirected, simple graphs. Given a graph~$G$, we write
$V(G)$ for its vertices and~$E(G)$ for its edges. Two vertices~$u,v\in V(G)$ are adjacent or neighbors if~$\{u,v\}\in E(G)$. The degree~$d_G(v)$ of a
vertex~$v\in V(G)$ is its number of neighbors in~$G$. 
We write~$N(v)$ for the set of neighbors and~$N[v]$ for the closed 
neighborhood~$N(v)\cup\{v\}$ of~$v$. For~$A\subseteq V(G)$, we 
write~$N[A]$ for~$\bigcup_{v\in A}N[v]$. We let~$N^1[v]:=N[v]$ and 
$N^{i+1}[v]:=N[N^i[v]]$ for~$i>1$. If $E'\subseteq E$, we write 
$N_{E'}(v)$ for the set $\{u \in V(G) : \{u,v\}\in E'\}$. A graph $G$ has radius at most~$r$ if 
there is a vertex $v\in V(G)$ such that 
$N^r[v]=V(G)$. 
The \emph{arboricity} of $G$ is the minimum number of forests into 
which its edges can be partitioned. A graph~$H$ is a subgraph of a 
graph~$G$ if~$V(H)\subseteq V(G)$ and~$E(H)\subseteq E(G)$. 
The edge density of~$G$ is the ratio $|E(G)|/|V(G)|$. It is well 
known that the arboricity of a graph is within factor $2$ of its 
\emph{degeneracy}, that is, $\max_{H\subseteq G}|E(H)|/|V(H)|$.
For~$A\subseteq V(G)$, the graph
$G[A]$ induced by~$A$ is the graph with vertex set~$A$ and edge set
$\{\{u,v\}\in E(G) : u,v\in A\}$. For~$B\subseteq V(G)$ we write~$G-B$
for the graph~$G[V(G)\setminus B]$. 

\smallskip
\noindent\textbf{Bounded depth minors and locally embeddable graphs.}
A graph~$H$ is a minor of a graph~$G$, written~$H\minor G$, if there is 
a set~$\{G_v : v\in V(H)\}$ of pairwise vertex disjoint and
connected subgraphs 
$G_v\subseteq G$ such that if~$\{u,v\}\in E(H)$, then there is an edge 
between a vertex of~$G_u$ and a vertex of~$G_v$. We say that~$G_v$ is
\emph{contracted} to the vertex~$v$. If $G_1,\ldots, G_k\subseteq V(G)$
are pairwise vertex disjoint and connected subgraphs of $G$, then we write 
$G/G_1/\ldots/G_k$ for the minor obtained by contracting the subgraphs~$G_i$ (observe that the order of contraction does not matter as the 
$G_i$'s are vertex disjoint). We call the set $\{G_v : v\in V(H)\}$ a
\emph{minor model} of $H$ in~$G$. We say that two minor models
$\{G^1_v : v\in V(H)\}$ and $\{G^2_v : v\in V(H)\}$ of $H$ in 
a graph $G$ disjoint if the sets $\bigcup_{v\in V(H)} V(G^1_v)$ and
$\bigcup_{v\in V(H)} V(G^2_v)$ are disjoint. 

A star is a connected graph~$G$ such that at most one 
vertex of~$G$, called the center of the star, has degree
greater than one. A graph~$H$ is a \emph{depth-$\mathit{1}$ minor} of~$G$ if~$H$ is 
obtained from a subgraph of~$G$ by star contractions, that is, if 
there is a set~$\{G_v : v\in V(H)\}$ of pairwise vertex disjoint 
stars~$G_v\subseteq G$ such that if~$\{u,v\}\in E(H)$, then 
there is an edge between a vertex of~$G_u$ and a vertex of~$G_v$.

More generally, for a non-negative integer $r$, a graph 
$H$ is a \emph{depth-$r$ minor} of $G$, written
$H\minor_r G$, if there is a set~$\{G_v : v\in V(H)\}$ of pairwise 
vertex disjoint connected subgraphs 
$G_v\subseteq G$ of radius at most $r$ such that if~$\{u,v\}\in E(H)$, 
then there is an edge between a vertex of~$G_u$ and a vertex of~$G_v$.

We write~$K_{t,3}$ for the complete bipartite
graph with partitions of size~$t$ and~$3$, respectively. 
A graph~$G$ is a \emph{locally embeddable graph} if it
excludes~$K_{3,t}$ as a \depthone minor for some~$t\ge 3$
and if $|E(H)|/|V(H)|\leq c$ for some constant $c$ and
all \depthone minors $H$ of $G$. 

More generally, we write $\nabla_r(G)$ for $\max_{H\minor_r G}|E(H)|/|V(H)|$. 
A class $\CCC$ of graphs has \emph{bounded expansion} if there is a function 
$f:\N\rightarrow\N$ such that $\nabla_r(G)\leq f(r)$ for all graphs $G\in \CCC$. 
This is equivalent to demanding that the arboricity of each depth-$r$ minor
of $G$ is functionally bounded by~$r$. The class $\CCC$ has \emph{sub-logarithmic 
expansion} if the bounding function $f(r)\in o\,(\log r)$. Note that if every graph
$G\in \CCC$ excludes a fixed minor, then $\CCC$ has constant expansion, hence
classes of sub-logarithmic expansion generalize proper minor closed classes of graphs. We refer to Figure~\ref{fig:classes} for the inclusion between 
the above defined classes. 

\pagebreak
\smallskip\noindent\textbf{Bounded genus graphs.} 
The (orientable, resp.\ non-orientable) genus of a graph is the 
minimal number~$\ell$ such that the graph can be embedded on 
an (orientable, resp.\ non-orientable) surface of genus~$\ell$. 
We write~$g(G)$ for the orientable genus of~$G$ and~$\tilde{g}(G)$ for the 
non-orientable genus of~$G$. Every connected planar 
graph has orientable genus~$0$ and non-orientable genus~$1$. 
In general, for connected~$G$, we have~$\tilde{g}(G)\leq 2g(G)+1$. 
On the other hand, there is no bound for~$g(G)$ in terms of~$\tilde{g}(G)$.
As all our results apply to both variants, for ease of presentation,
and as usual in the literature, we will simply speak of the genus of a graph
in the following.
We do not make explicit use of any topological arguments and hence refer
to~\cite{graphsurface} for more background on graphs on 
surfaces. We will use the following facts about bounded genus graphs. 

\smallskip
The first lemma states that graphs of genus~$g$ are closed under
taking minors. 
\begin{lemma}\label{lem:closureminor}
If~$H\minor G$, then~$g(H)\leq g(G)$ and~$\tilde{g}(H)\leq \tilde{g}(G)$. 
\end{lemma}

One of the arguments we will use is based on the fact that bounded 
genus graphs exclude large bipartite graphs as minors. The lemma follows 
immediately from \cref{lem:closureminor} and from the fact that 
$g(K_{m,n})=\left\lceil \frac{(m-2)(n-2)}{4}\right\rceil$ and 
$\tilde{g}(K_{m,n})=\left\lceil \frac{(m-2)(n-2)}{2}\right\rceil$ 
(see e.g.\ Theorem~4.4.7 in~\cite{graphsurface}). 

\begin{lemma}\label{lem:exclude}
If~$g(G)=g$, then~$G$ excludes~$K_{3,4g+3}$ as a minor 
and if~$\tilde{g}(G)=\tilde{g}$, then~$G$ excludes~$K_{3,2\tilde{g}+3}$ as a minor. 
\end{lemma}

Graphs of bounded genus do not contain 
many disjoint copies of minor models of~$K_{3,3}$: 
this is a simple consequence of the fact that the orientable 
genus of a connected graph is equal to the sum of the genera 
of its blocks (maximal connected subgraphs without a cut-vertex)
and a similar statement holds for the non-orientable genus, 
see Theorem~4.4.2 and Theorem~4.4.3 in~\cite{graphsurface}.

\begin{lemma}
\label{lem:decreasegenus}
A graph $G$ contains at most~$\max\{g(G), 2\tilde{g}(G)\}$ disjoint 
copies of minor models of~$K_{3,3}$. 
\end{lemma}

Finally, note that graphs of bounded genus have small 
edge density. It is straightforward to obtain the following from the
generalized Euler formula $n-e+f\leq \chi(G)$~\cite{graphsurface} for
example see~\cite{saeedthesis}. 

\begin{lemma}\label{lem:dens}
Every graph with at least $3$ vertices satisfies $|E(G)| \leq 3 \cdot |V(G)| + 6 g(G) - 6$ 
and~$|E(G)| \leq 3 \cdot |V(G)| + 3 \tilde{g}(G) - 3$.
\end{lemma}

\begin{lemma}\label{lem:degeneracy}
Let $\mathcal{G}$ be a class of graphs of genus at most $g$. Then the
degeneracy and edge density of every graph $G\in \mathcal{G}$ is bounded
by $5\sqrt{g}$.
\end{lemma}
\begin{proof}
Recall that the degeneracy of a graph $G$ is defined as 
$\max_{H\subseteq G}|E(H)|/|V(H)|$, which in particular
bounds the edge density $|E(G)|/|V(G)|$. It hence suffices
to bound the degeneracy of $G$. 
If $g=0$ the claim trivially holds, as in this case $G$ is 
planar and hence $\max_{H\subseteq G}|E(H)|/|V(H)|\leq
3$ by \cref{lem:closureminor} and \cref{lem:dens}.

Now assume $g\geq 1$ (we prove the lemma for graphs with 
orientable genus $g$, the proof for graphs of non-orientable genus 
$g$ is analogous). We fix any subgraph $H\subseteq G$.
We may assume that $H$ has at least $5\sqrt{g}$
vertices, otherwise, the statement is trivially true (as in this
case every vertex of $H$ has degree (in $H$) less than $5\sqrt{g}$). By 
\cref{lem:closureminor} and \cref{lem:dens}, 
we have $|E(H)|\leq 3\cdot |V(H)| + 6g(H) - 6\leq 3\cdot |V(H)| + 
6g$. This implies $|E(H)|/|V(H)|\leq 3+6g/|V(H)|\leq 3+6g/(5\sqrt{g})
\leq 5\sqrt{g}$, as claimed. 
\end{proof}

As an immediate corollary from  \cref{lem:closureminor}, 
 \cref{lem:exclude} and  \cref{lem:dens}, we get that if 
$\mathcal{G}$ is a class of graphs of bounded genus, then 
$\mathcal{G}$ is a class
of locally embeddable graphs. 

\medskip
\noindent\textbf{Dominating sets.}
Let~$G$ be a graph. A set~$D\subseteq V(G)$ \emph{dominates}~$G$ if all vertices of
$G$ lie either in~$D$ or are adjacent to a vertex of~$D$, that is, if~$N[D]=V(G)$. 
A minimum
dominating set~$D$ is a dominating set of minimum cardinality 
(among all dominating sets). The size of a minimum dominating set of~$G$ 
is denoted~$\gamma(G)$. 

\noindent\textbf{$f$-Approximation.} Let~$f:\mathbb{N}\rightarrow\mathbb{R}^+$.
Given an~$n$-vertex graph~$G$ and a set~$D\subseteq V(G)$, we say that~$D$ is 
an~$f$-approximation for the dominating set problem, if~$D$ is a dominating 
set of~$G$ and~$|D| \leq f(n)\cdot \gamma(G)$. An algorithm computes an 
$f$-approximation for the dominating set problem on a class~$\CCC$ of graphs 
if for all~$G\in\CCC$ it computes a set~$D$ which is an~$f$-approximation 
for the dominating set problem. 
If~$f$ maps every number to a fixed constant~$c$, we speak of a constant factor approximation.

\smallskip\noindent\textbf{Distributed complexity.}
We consider the standard \textit{LOCAL} model of distributed
computing~\cite{Linial:1992:LDG:130563.130578}, 
see also~\cite{local-survey} for a recent survey. 
A distributed system is modeled as a graph~$G$. 
At each vertex~$v\in V(G)$ there is an independent 
agent/host/processor with a unique identifier~$\mathit{id}(v)$. 
Initially, each agent has no knowledge about the network, 
but only knows its own identifier. Information about other agents 
can be obtained through message passing, i.e.,
 through repeated interactions with neighboring vertices, 
 which happens in synchronous communication rounds. 
 In each round the following operations are performed:
 \begin{enumerate}
 \item[(1)] Each vertex performs a local computation 
(based on information obtained in previous rounds).
\item[(2)] Each vertex~$v$ sends one message to
 each of its neighbors.
 \item[(3)] Each vertex~$v$ receives one message from 
each of its neighbors. 
 \end{enumerate}

The \emph{distributed complexity} 
of the algorithm is defined as the number of
 communication rounds until all agents terminate. 
 We call a distributed algorithm $r$-local, if its 
 output depends only on the~$r$-neighborhoods 
~$N^r[v]$ of its vertices. Observe that an 
$r$-local algorithm can (trivially) be implemented in 
$r$ rounds in the \textit{LOCAL} model.

\section{A Constant Local MDS Approximation}\label{sec:local-approx}

\noindent Let us start by revisiting the MDS approximation algorithm for planar graphs by Lenzen et al.~\cite{ds-planar}, see Algorithm~\ref{alg:lenzen}.
The algorithm works in two phases. In the first phase, 
it adds all vertices whose (open) neighborhood cannot 
be dominated by a small number of vertices (to be precise, 
by at most~$6$ vertices) to a set~$D$. 
It has been shown in~\cite{ds-planar} that the set~$D$ is small 
(at most $4$ times larger than a minimum dominating set) in planar graphs. 
In the second phase, the algorithm defines a dominator function~$dom$ which 
maps every vertex~$v$ that is not dominated yet by~$D$ to its dominator. 
The dominator~$dom(v)$ of~$v$ is chosen arbitrary among 
those vertices of~$N[v]$ which dominate the 
maximal number of vertices not dominated yet.

\begin{algorithm}[ht]
\caption{~~Dominating Set Approximation Algorithm for Planar Graphs~\cite{ds-planar}}
\begin{algorithmic}[1]
\vspace{2mm}
\STATE Input: Planar graph~$G$

\medskip
\STATE~$(*$ \emph{Phase 1} ~$*)$
\STATE~$D \gets \emptyset$
\STATE \textbf{for}~$v\in V$ (in parallel) \textbf{do}
\STATE \qquad\textbf{if} there does not exist a set~$A\subseteq V(G)\setminus \{v\}$ such that~$N(v)\subseteq N[A]$ and~$|A|\leq 6$ \textbf{then}
\STATE \qquad \qquad $D\gets D\cup \{v\}$
\STATE \qquad \textbf{end if}
\STATE \textbf{end for}

\medskip
\STATE~$(*$ \emph{Phase 2} ~$*)$
\STATE~$D'\gets \emptyset$
\STATE \textbf{for}~$v\in V$ (in parallel) \textbf{do}
\STATE \qquad $d_{G-D}(v)\gets |N[v]\setminus N[D]|$
\STATE \qquad \textbf{if}~$v\in V\setminus N[D]$ \textbf{then}

\STATE \qquad \qquad $\Delta_{G-D}(v)\gets \max_{w\in N[v]}d_{G-D}(w)$
\STATE \qquad \qquad choose any~$dom(v)$ from $N[v]$ with
$d_{G-D}(dom(v))=\Delta_{G-D}(v)$
\STATE \qquad \qquad $D'\gets D'\cup \{dom(v)\}$
\STATE \qquad \textbf{end if}
\STATE \textbf{end for}
\STATE \textbf{return}~$D\cup D'$
\end{algorithmic}\label{alg:lenzen}
\end{algorithm}

We now propose the following small change to the algorithm. 
As additional input, we require an integer~$c$ which bounds the edge
density of \depthone minors of $G$
and we replace the condition~$|A|\leq 6$ in Line~5 by the 
condition~$|A|\leq 2c$. In the rest of this section, we show that 
the modified algorithm computes a constant factor approximation on 
any locally embeddable class of graphs. Note that the algorithm 
does not have to compute the edge density of~$G$, which is not 
possible in a local manner. Rather, we leverage  \cref{lem:dens} 
which upper bounds the edge density for any fixed 
class of bounded genus graphs: this upper bound
can be used as an input to the local algorithm. 

We first show that the set~$D$ computed in Phase~1 
of the algorithm is small. The following lemma is a straightforward 
generalization of Lemma~6.3 of~\cite{ds-planar}, 
which in fact does not use topological arguments at all. 

\begin{lemma}\label{thm:largeneighbourhood}
  Let~$G$ be a graph and let~$M$ be 
  a minimum dominating set of~$G$. 
  Assume that for some constant~$c$ all \depthone
  minors~$H$ of~$G$ satisfy~$|E(H)|/|V(H)|\leq c$. Let 
  \begin{align*} D\coloneqq & \{v\in V(G)~:~\text{there is no set
  $A\subseteq V(G)\setminus\{v\}$  such that $N(v)\subseteq N[A]$ and $|A|\leq 2c$}\}.\end{align*}
 Then~$|D|\leq (c+1)\cdot |M|$.
\end{lemma}
\begin{proof}
  Let~$H$ be the induced subgraph of $G$ with~$V(H)=M\cup N[D\setminus M]$. Since $M$ is a dominating set, we can fix for each
  $v\in N[D\setminus M]\setminus (D\cup M)$ a vertex $m_v\in M$ that
  is adjacent to $v$. Then for each $m\in M$, the subgraph $G_m$
  which consists of the central vertex $m$ and all $v\in N[D\setminus M]
  \setminus (D\cup M)$ such that $m=m_v$ and all edges $\{m,v\}$
  is a star. Furthermore, observe that for different $m_1,m_2\in M$
  the starts $G_{m_1}$ and $G_{m_2}$ are vertex disjoint. 

\pagebreak
  We construct a \depthone minor~$\tilde{H}$ of~$H$ by contracting the
  star subgraphs~$G_m$ for~$m\in M$ into vertices
 $v_m$. Then (all non-trivial inequalities will be explained below)
  \begin{align}
  (c+1)\cdot |D\setminus M| 
  & =  (2c+1)\cdot |D\setminus M|-c\cdot |D\setminus M| \nonumber\\
    & \leq  \sum_{w\in D\setminus M} d_{\tilde{H}}(w)-|E(\tilde{H}[D\setminus M])|\\
    & \leq  |E(\tilde{H})|\\
    & \leq  c\cdot |V(\tilde{H})|\\
    & = c\cdot (|D\setminus M|+ |M|),
  \end{align}
  and hence~$|D\setminus M|\leq c\cdot|M|$, which implies the claim. 

\bigskip  
 \begin{enumerate}
 \item Let~$w\in D\setminus M$. As~$N_G(w)$ cannot be covered by
  less than~$(2c+1)$ elements from~$V(G)\setminus \{w\}$ (by definition of~$D$),~$w$
  also has at least~$(2c+1)$ neighbors in~$\tilde{H}$. Hence 
  $\sum_{w\in D\setminus M} d_{\tilde{H}}(w)\geq (2c+1)\cdot |D\setminus M|$. On the other hand, every subgraph $\tilde{H}'$ of 
  $\tilde{H}$ has at most~$c\cdot |V(\tilde{H}')|$ edges (every 
  subgraph of a \depthone minor is also a \depthone minor of $G$
  and we assume that every \depthone minor of $G$ has edge 
  density at most $c$). Hence $\tilde{H}[D\setminus M]$ has at most 
 $c\cdot |D\setminus M|$ edges.
 \item Every edge $\{v,w\}\in \tilde{H}$ with $v,w\in D\setminus M$
 is counted twice in the sum $\sum_{w\in D\setminus M} d_{\tilde{H}}(w)$, 
 once when we count $d_{\tilde{H}}(v)$ and once when counting
 $d_{\tilde{H}}(w)$. By subtracting the number of edges that run
 between vertices of $D\setminus M$ we get the second inequality. 
 \item The third inequality holds by assumption on the density of
 \depthone minors of $G$. 
 \item By construction, all vertices of $N[D\setminus M]\setminus D$ 
 disappear into some star $G_m$, hence $\tilde{H}$ has exactly 
 $|D\setminus M|+|M|$ vertices. 
 \end{enumerate}

\end{proof}

\begin{assumption}
For the rest of this section, we fix a graph $G$ which is 
locally embeddable,  that is, $G$ excludes $K_{3,t}$ for some $t$ 
as \depthone minor
and all \depthone minors~$H$ of~$G$ satisfy $|E(H)|/|V(H)|\leq c$ 
for some constant~$c$ (hence,  \cref{thm:largeneighbourhood}
can be applied). Furthermore, we fix~$M$ and~$D$ as in 
 \cref{thm:largeneighbourhood}. 
\end{assumption}

Let us write~$R$ for the set~$V(G)\setminus N[D]$ of 
vertices which are not dominated by~$D$. The algorithm 
defines a dominator function~$dom:R\rightarrow N[R]\subseteq 
V(G)\setminus D$. The set~$D'$ computed by the algorithm
is the image~$dom(R)$, which is a dominating set of vertices in~$R$. As~$R$ contains 
the vertices which are not dominated by~$D$,~$D'\cup D$ is a 
dominating set of~$G$. This simple observation proves that the 
algorithm correctly computes a dominating set of~$G$. 
Our aim is to find a bound on~$|dom(R)|$.

\smallskip
We fix an ordering of~$M$ as~$m_1,\ldots, m_{|M|}$ such that 
the vertices of~$M\cap D$ are first (minimal) in the ordering and 
inductively define a minimal set~$E'\subseteq E(G)$ such
 that~$M$ is also a dominating set with respect to~$E'$ as follows. For $i=1$, we add all 
 edges~$\{m_1,v\}\in E(G)$ with~$v\in N(m_1)\setminus M$ to~$E'$. 
 If for some $i\geq 1$ we have defined the set of edges $E'$ which 
 are incident with $m_1,\ldots, m_i$, we continue to add for $i+1$
 all edges~$\{m_{i+1}, v\}\in E(G)$ with~$v\in N(m_{i+1})
 \setminus(M\cup N_{E'}(\{m_1,\ldots, m_{i}\}))$. 

\smallskip
For~$m\in M$, let~$G_m$ be the star subgraph of~$G$ with 
center~$m$ and all vertices~$v$ with~$\{m,v\}\in E'$. Let~$H$
be the \depthone minor of~$G$ which is obtained by contracting all 
stars~$G_m$ for~$m\in M$. This construction is visualized in 
Figure~\ref{fig:construction}. In the figure, solid (undirected) lines
represent edges from~$E'$, edges incident with $m\in M$ which 
are not in~$E'$ are dashed. We want to count 
the \emph{endpoints} of directed edges, 
which represent the dominator function~$dom$. 

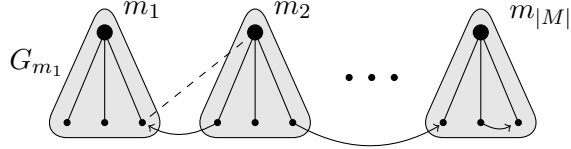
\begin{figure}[h!]
\centering

\begin{tikzpicture}

\filldraw[draw=black,fill=black!10!white,rounded corners=10pt] (1,1) -- (0.15, -0.9) -- (1.85,-0.9) -- cycle;

\draw[fill=black] (1, 0.5) circle (1mm);
\node at (1.5, 0.8) {$m_1$};
\node at (0.1, 0.1) {$G_{m_1}$};

\foreach \x in {0.5, 1, 1.5}
{
	\draw[fill=black] (\x, -0.7) circle (0.4mm);
	\draw (\x,-0.7) -- (1,0.5);
}


\filldraw[draw=black,fill=black!10!white,rounded corners=10pt] (3,1) -- (2.15, -0.9) -- (3.85,-0.9) -- cycle;

\draw[fill=black] (3, 0.5) circle (1mm);
\node at (3.5, 0.8) {$m_2$};

\foreach \x in {2.5, 3, 3.5}
{
	\draw[fill=black] (\x, -0.7) circle (0.4mm);
	\draw (\x,-0.7) -- (3,0.5);
}

\draw[fill=black] (4.25, -0.1) circle (0.4mm);
\draw[fill=black] (4.55, -0.1) circle (0.4mm);
\draw[fill=black] (4.85, -0.1) circle (0.4mm);

\draw[dashed] (3,0.5) -- (1.5, -0.7);


\filldraw[draw=black,fill=black!10!white,rounded corners=10pt] (6,1) -- (5.15, -0.9) -- (6.85,-0.9) -- cycle;

\draw[fill=black] (6, 0.5) circle (1mm);
\node at (6.8, 0.7) {$m_{|M|}$};

\foreach \x in {5.5, 6, 6.5}
{
	\draw[fill=black] (\x, -0.7) circle (0.4mm);
	\draw (\x,-0.7) -- (6,0.5);
}

\draw (2.5,-0.7) edge[out=210,in=330,->] (1.58, -0.74);
\draw (5.45,-0.74) edge[out=210,in=330,<-] (3.5, -0.7);
\draw (6.4,-0.74) edge[out=210,in=330,<-] (6, -0.7);

\end{tikzpicture}
\caption{The graphs $G_m$. Solid (undirected) lines represent edges from $E'$, directed edges represent the dominator function $dom$. Dashed lines represent edges incident with $m\in M$ which are not in $E'$.}
\label{fig:construction}
\end{figure}
In the following, we call a directed edge which
represents the function~$dom$ a \emph{$dom$-edge}. We did not 
draw~$dom$-edges that either start or end in~$M$. When counting~$|dom(R)|$, 
we may simply add a term~$2|M|$ to estimate the number of endpoints of 
those edges. 
We also did not draw a~$dom$-edge starting in~$G_{m_1}$. 
In the figure, we assume that the vertex~$m_1$ belongs to 
$M\cap D$. Hence every vertex~$v$ from~$N[m_1]$ is 
dominated by a vertex from~$D$ and the function is thus not 
defined on such~$v$. However, the vertices of $N(m_1)$ may
still serve as dominators, as shown in the figure. 

\smallskip
The graph $H$ has~$|M|$ vertices and by our 
assumption on the 
density of \depthone minors of~$G$, 
it has at most~$c\cdot |M|$ edges.

\smallskip
Our analysis proceeds as follows. We distinguish between
two types of~$dom$-edges, namely those which go from 
one star to another star and those which start and end in the 
same star. By the star contraction, all edges which go 
from one star to another star are represented by a 
single edge in~$H$. We show in  \cref{lem:edgerepresentative} 
that each edge in~$H$ does not represent many such~$dom$-edges 
with distinct endpoints. As~$H$ has at most~$c\cdot |M|$ 
edges, we will end up with a number of such edges 
that is linear in~$|M|$. On the other hand, all edges 
which start and end in the same star completely 
disappear in~$H$. In  \cref{lem:insidestars}
we show that these star contractions ``absorb'' only few 
such edges with distinct endpoints.

\smallskip
We first show that an edge in~$H$ represents only 
few~$dom$-edges with distinct endpoints. For each
$m\in M\setminus D$, we fix a set~$C_m\subseteq V(G)\setminus \{m\}$ 
of size at most~$2c$ which dominates~$N_{E'}(m)$. The existence of
such a set follows from the definition of the set $D$. Recall that we 
assume that~$G$ excludes~$K_{3,t}$ as \depthone minor. 

\begin{lemma}\label{lem:edgerepresentative}
Let~$1\leq i<j\leq |M|$. Let ~$N_i:=N_{E'}(m_i)$ and~$N_j:=N_{E'}(m_j)$. 
\begin{enumerate}
\item If $m_j\in M\setminus D$, then \[|\{u \in N_j: \text{ there is~$v\in N_i$ with~$\{u,v\}\in E(G)\}|\leq 2ct$}.\]
\item If~$m_i\in M\setminus D$ (and hence~$m_j\in M\setminus D$), then \[|\{u \in N_i: \text{ there is~$v\in N_j$ with~$\{u,v\}\in E(G)\}|\leq 4ct$}.\]
\end{enumerate}
\end{lemma}
\smallskip
\begin{proof}
By definition of~$E'$, we may assume that $m_i\not\in C_{m_j}$
($m_i$ is not connected to $N_{E'}(m_j)$ and hence it can be safely 
removed if it appears in $C_{m_j}$).
Let~$c\in C_{m_j}$ be arbitrary. Then there are at most~$t-1$ 
distinct vertices~$u_1,\ldots, u_{t-1}\in (N_j\cap N(c))$ 
such that there are~$v_1,\ldots, v_{t-1}\in N_i$ 
(possibly not distinct) with~$\{u_k, v_k\}\in E(G)$ for
all~$k$,~$1\leq k\leq t-1$. 
Otherwise, we can contract the star with center~$m_i$
and branch vertices~$N(m_i)\setminus \{c\}$ and
thereby find~$K_{3,t}$ as  \depthone minor, a contradiction. 
See Figure~\ref{fig:contraction} for an illustration in the
case of an excluded~$K_{3,3}$. 
Possibly,~$c\in N_j$ and
 it is connected to a vertex of~$N_i$, hence we have 
at most~$t$ vertices in~$N_j\cap N[c]$ with a connection to 
$N_i$. As~$|C_{m_j}|\leq 2c$, we conclude the first item.
 
Regarding the second item, let~$c\in C_{m_i}$ be arbitrary. If~$c\neq m_j$, 
we conclude just as above, that there are at most~$t-1$ distinct 
vertices~$u_1,\ldots, u_{t-1}\in (N_i\cap N(c))$ such that there 
are~$v_1,\ldots, v_{t-1}\in N_j$ (possibly not distinct) with 
$\{u_k, v_k\}\in E(G)$ for all~$k$,~$1\leq k\leq t-1$ and 
hence at most~$t$ vertices in~$N_i\cap N[c]$ with a 
connection to~$N_j$. Now assume~$c=m_j$. Let~$c'\in C_{m_j}$. 
There are at most~$t-1$ distinct vertices~$u_1,
\ldots, u_{t-1}\in (N_i\cap N_E(m_j))$ such that there 
are vertices~$v_1,\ldots, v_{t-1}\in N_j\cap N(c)$ (possibly not distinct) 
with~$\{u_k, v_k\}\in E(G)$ for all~$k$,~$1\leq k\leq t-1$. 
Again, considering the possibility that~$c'\in N_i$, there 
are at most~$t$ vertices in~$N_i\cap N_E(m_j)$ with a 
connection to~$N_j\cap N(c)$. As~$|C_{m_j}|\leq 2c$, 
we conclude that in total there are at most~$2ct$ vertices 
in~$N_i\cap N_E(m_j)$ with a connection to~$N_j$. 
In total, there are hence at most~$(2c-1)t + 2ct\leq 4ct$ 
vertices of the described form.  
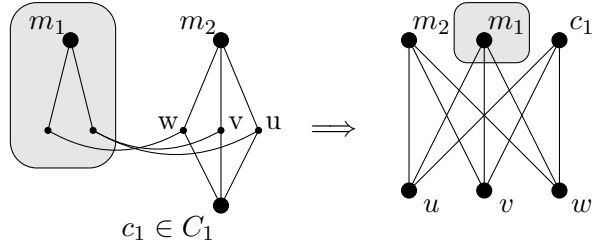
\begin{figure}[h!]
\centering

\begin{tikzpicture}

\begin{scope}[xshift=4cm]
\begin{scope}[yscale=1,xscale=-1]
\draw[fill=black] (1, 0.5) circle (1mm);

\draw[fill=black] (1, -1.7) circle (1mm);
\node at (1.7, -2) {$c_1\in C_1$};

\foreach \x in {0.5, 1, 1.5}
{
	\draw[fill=black] (\x, -0.7) circle (0.4mm);
	\draw (\x,-0.7) -- (1,-1.7);
	\draw (\x,-0.7) -- (1,0.5);
}
\node at (0.3, -0.6) {u};
\node at (0.8, -0.6) {v};
\node at (1.7, -0.6) {w};

\filldraw[draw=black,fill=black!10!white,rounded corners=10pt] (2.4,1) rectangle (3.8,-1.2);

\draw[fill=black] (3, 0.5) circle (1mm);
\node at (1.3, 0.7) {$m_2$};
\node at (3.3, 0.7) {$m_1$};

\foreach \x in {2.7, 3.3}
{
	\draw[fill=black] (\x, -0.7) circle (0.4mm);
	\draw (\x,-0.7) -- (3,0.5);
}

\draw (2.7,-0.7) edge[out=210,in=330,-] (0.5, -0.7);
\draw (2.7,-0.7) edge[out=210,in=330,-] (1, -0.7);
\draw (3.3,-0.7) edge[out=210,in=330,-] (1.5, -0.7);

\end{scope}
\end{scope}

\node at (4.5, -0.7) {$\Longrightarrow$};

\begin{scope}[xshift=5.5cm]

\draw[fill=black] (0, 0.5) circle (1mm);
\node at (0.3, 0.7) {$m_2$};

\filldraw[draw=black,fill=black!10!white,rounded corners=5pt] (0.6,1) rectangle (1.6,0.2);
\draw[fill=black] (1, 0.5) circle (1mm);
\node at (1.3, 0.7) {$m_1$};

\draw[fill=black] (2, 0.5) circle (1mm);
\node at (2.3, 0.7) {$c_1$};

\draw[fill=black] (0, -1.5) circle (1mm);
\node at (0.3, -1.7) {$u$};

\draw[fill=black] (1, -1.5) circle (1mm);
\node at (1.3, -1.7) {$v$};

\draw[fill=black] (2, -1.5) circle (1mm);
\node at (2.3, -1.7) {$w$};

\foreach \x in {0, 1, 2}
{
	\foreach \y in {0,1,2}
	{
		\draw (\x,0.5) -- (\y, -1.5);
	}
	
}

\end{scope}

\end{tikzpicture}
\caption{Visualisation of the proof of  \cref{lem:edgerepresentative} in the case of excluded $K_{3,3}$}
\label{fig:contraction}
\end{figure}
\end{proof}

We write~$Y$ for the set of all vertices~$\{u\in N_{E'}(m_i) :$ 
$m_i\not\in D$ and there is~$v\in N_{E'}(m_j)$,~$j\neq i$ 
and~$\{u,v\}\in E(G)\}$. 

\begin{corollary}\label{crl:numedgesbetweendiamonds}
$|Y|\leq 6c^2t\cdot |M|$. 
\end{corollary}
\begin{proof}
Each of the~$c\cdot |M|$ many edges in~$H$ represents edges 
between~$N_i$ and~$N_j$, where~$N_i$ and~$N_j$ are defined 
as above. By the previous lemma, if~$i<j$, there are at 
most~$2ct$ vertices in~$N_i\cap Y$ and at most~$4ct$ 
vertices in~$N_j\cap Y$. Hence in total, each edge accounts 
for at most~$6ct$ vertices in~$Y$.   
\end{proof}

We continue to count the edges which are inside the stars. First, 
we show that every vertex has small degree inside its own star. 

\begin{lemma}\label{lem:edgestosamestar}
Let~$m\in M\setminus D$ and let~$v\in N_{E'}(m)\setminus C_m$. Then \[|\{u \in N_{E'}(m) : \{u,v\}\in E(G)\}|\leq 2c(t-1).\] 
\end{lemma}
\smallskip
\begin{proof}
Let~$c\in C_m$. By the same argument as in 
 \cref{lem:edgerepresentative}, there are 
at most~$t-1$ distinct vertices~$u_1,\ldots, u_{t-1}\in 
(N_{E'}(m)\cap N(c))$ such that~$\{u_k, v\}\in E(G)$ 
for all~$k$,~$1\leq k\leq t-1$.  
\end{proof}

Let~$C\coloneqq \bigcup_{m\in M\setminus D}C_m$. We show that 
there are 
only few vertices which are highly connected to~$M\cup C$. 
Let~$Z:=\{u \in N_{E'}(M\setminus D) : |N(u)\cap (M\cup C)|>4c\}$. 

\begin{lemma}\label{lem:Z}
\[|Z|< |M\cup C|.\]
\end{lemma}
\smallskip
\begin{proof}
Assume that~$|Z|> |M\cup C|$. Then the subgraph 
induced by~$Z\cup M\cup C$ has more than 
$\frac{1}{2}4c|Z|$ edges and~$|Z\cup M\cup C|$ 
vertices. Hence its edge density is larger than 
$2c|Z|/(|Z\cup M\cup C|)> 2c|Z|/(2|Z|)= c$, 
contradicting our assumption on the edge density of 
\depthone minors of~$G$ (which includes its subgraphs).
\end{proof}

Finally, we consider the image of the~$dom$-function inside the stars. 

\begin{lemma}\label{lem:insidestars}
\[\left|\bigcup_{m\in M\setminus D}\{u\in N_{E'}(m) : 
dom(u)\in (N_{E'}(m)\setminus (Y\cup Z))\}\right|\]
 \[ \leq 
(2(t-1)+4)c|M|.\]
\end{lemma}
\smallskip
\begin{proof}
Fix some~$m\in M\setminus D$ and some
$u\in N_{E'}(m)$ with~$dom(u)\in N_{E'}(m)\setminus
(Y\cup Z)$. Because~$dom(u)\not\in Y$,~$dom(u)$ is not
connected to a vertex of a different star, except possibly 
for vertices from~$M$. Because~$dom(u)\not\in Z$, it is 
however connected to at most~$4c$ vertices from~$M\cup C$. 
Hence it is connected to at most~$4c$ vertices 
from different stars. According to  \cref{lem:edgestosamestar}, 
$dom(u)$ is connected to at most~$2c(t-1)$ vertices from
the same star. Hence the degree of~$dom(u)$ is at most 
$4c+2c(t-1)$. Because~$u$ preferred to choose~$dom(u)\in
N_{E'}(m)$ over~$m$ as its dominator, we conclude that~$m$ 
has at most~$4c+2c(t-1)$~$E'$-neighbors. Hence, in total there can 
be at most~$(2(t-1)+4)c\cdot |M|$ such vertices.
\end{proof}

We are now ready to put together the numbers. 

\begin{lemma}\label{lem:mainlemma}
If all \depthone minors $H$ of $G$ have edge density at 
most~$c$ and $G$ excludes~$K_{3,t}$ 
as \depthone minor, then the modified algorithm computes 
a~$6c^2t+(2t+5)c+4$ approximation for the 
minimum dominating set problem on~$G$. 
\end{lemma}
\begin{proof}
Since~$M$ is 
a dominating set also with respect to the edges~$E'$,
it suffices to bound~$|\{ dom(u) : u\in (N_{E'}
[M\setminus D]\setminus N[D])\}|$. This set is partitioned
into the following (not necessarily disjoint) sets. First, all endpoints 
of $dom(R)$ that 
go from one star to another star are found in one of the 
sets $Y=\{u\in N_{E'}(m_i) :$ there is~$v\in N_{E'}(m_j)$, 
$i\neq j$ and~$\{u,v\}\in E(G)\}$, $dom(R)\cap M$ and $dom(M)$. 
All other dom-edges connect vertices inside individual stars. 
Here, $dom(R)$ splits into those vertices which are highly
connected to $M\cup C$, that is, the set 
$Z=\{u \in N_{E'}(M\setminus D) : 
|N(u)\cap (M\cup C)|>4c\}$, the set $C$ and the set 
$Y$ (which will not be counted twice though). All other 
dom-edges lead to vertices which lie neither in $Y$ nor in $Z$. 

In the previous lemmas we have bounded the sizes of 
each of the described sets. 
The set~$D$ has size at most~$(c+1)|M|$ according
to  \cref{thm:largeneighbourhood}. 
According to 
\cref{crl:numedgesbetweendiamonds}, the 
set~$Y$ has size at most 
$6c^2t|M|$. In particular, there are at most so 
many vertices~$dom(u)\in N_{E'}(m_i)$ with 
$u\in N_{E'}(m_j)$ for~$i\neq j$. Clearly, $|dom(R)\cap M|\leq |M|$ 
and $|dom(M)|\leq |M|$. 
According to 
 \cref{lem:Z}, the set~$Z$ satisfies~$|Z|< |M\cup C|$. 
We have~$|C|\leq 2c|M|$, as each~$C_m$ has size at 
most~$2c$. It remains to count the image of~$dom$ 
inside the stars which do not point to~$Y$ or~$Z$. 
According to  \cref{lem:insidestars}, this image has size 
at most~$(2(t-1)+4)c|M|$. In total, we can bound $|dom(R)|$ by 
\begin{align*}
 (c+1)|M| & +6c^2t|M|+2|M|+(2c+1)|M| + (2(t-1)+4)c|M|
 \leq  \; (6c^2t+(2t+5)c+4)|M|.
\end{align*}
\end{proof}

Our theorem for bounded genus graphs is now a corollary of 
 \cref{lem:exclude}, \ref{lem:degeneracy} and \ref{lem:mainlemma}. 


\begin{theorem}\label{thm:main}
Let~$\mathcal{C}$ be a class of graphs of orientable 
genus at most~$g$ (non-orientable genus at most~$\tilde{g}$ resp.). 
The modified algorithm computes an~$\Oof(g^2)$-approximation 
($\Oof(\tilde{g}^2)$-approximation resp.) for the dominating set 
in a constant number of communication rounds. 
\end{theorem}

For the special case of planar graphs, our analysis shows that the algorithm 
computes a~$199$-approximation. This is not much worse 
than Lenzen et al.'s original analysis (130), however, off by a factor 
of almost~$4$ from Wawrzyniak's~\cite{better-upper-planar} improved
analysis (52).

\subsection{Improving the Approximation Factor with Preprocessing}\label{sec:improved}

\noindent We now show the approximation 
factors related to the genus~$g$, derived
in the previous section, can be improved 
using a local preprocessing step.

Given a graph~$G$ and a vertex~$v\in V(G)$, let $K=\{K_1,\ldots,K_j\}$
denote the set of minimal subgraphs of $G$
containing $v$ such that for all $1\le i \le j$,
$K_{3,3}$ is a \depthone minor of $K_i$. Let $K_h\in K$ be the one
with lexicographically smallest identifiers in $K$. We call $K_h$ 
the $v$-canonical subgraph of $G$ and we denote it by $K_v$. If
$K=\emptyset$ we set $K_v:=\emptyset$.


\begin{lemma}\label{lem:findk33}
Given a graph~$G$ and a vertex~$v\in V(G)$. 
The $v$-canonical subgraph~$K_v$ of $v$ can be computed locally in
at most~$6$ communication rounds. Furthermore, $K_v$ has at most $24$ vertices.
\end{lemma}
\begin{proof}\label{alg:findk33}
The proof is constructive. As $K_{3,3}$ has diameter $2$, 
every minimal subgraph of $G$ containing $K_{3,3}$ as a 
\depthone minor has diameter at most $6$ (every edge
may have to be replaced by a path of length $3$). 
Therefore, it suffices to consider the subgraph induced by 
the vertices at distance at most $6$ from $v$, 
$H=G[N^6(v)]$, and find the lexicographically minimal 
subgraph which contains $K_{3,3}$ as \depthone minor
in~$H$ which includes~$v$ as a vertex. If this is the case, 
we output it as~$K_v$; otherwise we output the 
empty set. Furthermore,~$K_{3,3}$ has~$9$ 
edges and hence a minimal subgraph containing it as
\depthone minor has
at most~$24$ vertices (again, every edge is subdivided at 
most twice and $2\cdot 9+6=24$).
\end{proof}

To improve the approximation factor, we 
propose the following modified algorithm,
see Algorithm~\ref{alg:modified-approx}. We first carry out
the first phase of Algorithm~\ref{alg:lenzen} with density
parameter $10\sqrt{g}$ (the parameter is twice the edge
density of the input graph).  In the following preprocessing
phase we eliminate all copies of depth-$1$ minor models of $K_{3,3}$ that
$G$ possibly contains. By  \cref{lem:decreasegenus} 
we know that there are at most $g$ (where $g$ is the
genus of the graph) disjoint such models. As 
guaranteed by  \cref{lem:findk33}, the vertices can 
make a canonical local choice on which model the delete. 
After $g$ elimination rounds we are left with a locally embeddable
graph (with the parameter $t=3$)
and we call the second phase of Algorithm~\ref{alg:lenzen}.

\begin{algorithm}[t]
\caption{Dominating Set Approximation for Graphs of Genus~$\leq g$}
\label{alg:modified-approx}
\begin{algorithmic}[1]
\vspace{2mm}
\STATE Input: Graph~$G$ of genus at most $g$

\smallskip
\STATE \textbf{Run Phase~1 of Modified Algorithm~\ref{alg:lenzen} with density
parameter $10\sqrt{g}$ to obtain set $D$} 

\smallskip
\STATE~$(*$ \emph{Preprocessing Phase} ~$*)$

\STATE \textbf{for}~$v\in V-D$ (in parallel) \textbf{do}
\STATE \qquad \textbf{compute~$K_v$ in~$G-D$ (see  \cref{lem:findk33}})
\STATE \textbf{end for}

\STATE \textbf{for}~$i=1..g$ \textbf{do}
\STATE \qquad \textbf{for}~$v\in V-D$ (in parallel) \textbf{do}
\STATE \qquad \qquad \textbf{if}~$K_v\neq\emptyset$ \textbf{then} \textbf{chosen : = true}
\STATE \qquad \qquad \textbf{for all~$u\in N^{12}(v)$ do}
\STATE \qquad \qquad \qquad \textbf{if~$K_u\cap K_v \neq\emptyset$
 and~$u < v$}\textbf{ then chosen := false}
 \STATE \qquad \qquad \textbf{end for}
\STATE \qquad \qquad \textbf{if (chosen = true) then }$D := D\cup V(K_v)$
\STATE \qquad \textbf{end for}
\STATE \textbf{end for}
\smallskip
\STATE \textbf{Run Phase~2 of Algorithm~\ref{alg:lenzen}~}
\end{algorithmic}
\end{algorithm}

\begin{theorem}\label{thm:modified}
Algorithm~\ref{alg:modified-approx} provides a~$24g+\Oof(1)$ MDS approximation 
for graphs of genus 
at most~$g$, and requires~$12g+\Oof(1)$ communication rounds.
\end{theorem}
\begin{proof}
The resulting vertex set is clearly a dominating set.  It remains to
bound its size. 

As Phase~1 is unchanged, the computed
set $D$ is at most $6\sqrt{g}$ times larger than an 
optimal dominating set by  \cref{lem:dens}:
the algorithm is called with parameter $2c$ and outputs
a set at most $c+1$ times larger than an optimal dominating
set; here, $c=5\sqrt{g}$ according to  \cref{lem:degeneracy}. 

\smallskip
In the following 
preprocessing phase, if for two vertices~$u\neq v$ we choose
both~$K_u$ and $K_v$, then they must be disjoint:
Since the diameter of any \depthone minor of 
$K_{3,3}$ is at most~$6$, if two such canonical 
subgraphs $K_u$ and $K_v$  intersect, then 
the distance between  $u,v$ can be at most~$12$. Hence, 
each vertex $v$ can decide in its $12$-neighborhood
whether its canonical subgraph $K_v$ is the smallest among 
all choices. On the other hand, by 
 \cref{lem:decreasegenus}, there are at 
most~$g$ disjoint such minor models. So in
the \emph{preprocessing} phase, we can remove at
most $g$ disjoint subgraphs~$K_v$ (and add their vertices 
to the dominating set) and thereby 
select at most~$24g$ extra vertices for the dominating set.
Once the \emph{preprocessing} phase is finished, the remaining graph 
is locally embeddable. Observe that if the input graph $G$ is planar, no vertices
will be added to $D$ in the preprocessing phase. 

\smallskip
In order to compute the  size of the set in the third phase, we 
can use the analysis of  \cref{lem:mainlemma}
for~$t=3$,
 which together 
with the first phase and preprocessing phase,
results in a~$24g+\Oof(1)$-approximation guarantee.

\smallskip
To count the number of communication rounds, note that 
the only change happens in the second phase. In that phase, 
in each iteration, we need~$12$ communication rounds to 
compute the~$12$-neighborhood. Therefore, the number of 
communication rounds is~$12g + \Oof(1)$.
\end{proof}

This significantly improves the 
approximation upper bound of  \cref{thm:main}: namely from 
$4(6c^2+2c)g + \Oof(1)$, where $c=\Oof(\sqrt{g})$, 
hence from $\Oof(g^2)$ to~$24g + \Oof(1)$,
at the price of~$12g$ extra communication rounds.

\subsection{A Logical Perspective}\label{sec:stoneage}

\noindent Interestingly, as we will elaborate in the following, a small modification of 
Algorithm~\ref{alg:lenzen} can be interpreted both from a distributed computing
perspective, namely as a local algorithm
of constant distributed time complexity, as well as from a logical perspective.

First order logic has atomic formulas of the form 
$x=y, x<y$ and~$E(x,y)$, where~$x$ and~$y$ are 
first-order variables and $E$ is a binary relation symbol. 
The set of first order formulas is 
closed under Boolean combinations and existential and 
universal quantification over the vertices of a graph. 
To define the semantics, we inductively define a satisfaction 
relation~$\models$, where for a graph~$G$, a formula 
$\phi(x_1,\ldots, x_k)$ and vertices~$v_1,\ldots, v_k
\in V(G)$,~$G\models\phi(v_1,\ldots, v_k)$ means 
that~$G$ satisfies~$\phi$ if the free variables~$x_1,\ldots, 
x_k$ are interpreted as~$v_1,\ldots, v_k$, respectively. 
The free variables of a formula are those that have
an instance not in the 
scope of a quantifier, and we write~$\phi(x_1, 
\ldots , x_k)$ to indicate that the free variables
of the formula~$\phi$ are among~$x_1,\ldots, x_k$. 
For~$\phi(x_1,x_2)=x_1<x_2$, we have~$G\models
\phi(v_1,v_2)$ if~$v_1<v_2$ with respect to the 
ordering~$<$ of~$V(G)$ and for~$\phi(x_1,x_2)=
E(x_1,x_2)$ we have~$G\models\phi(v_1,v_2)$
if~$\{v_1,v_2\}\in E(G)$. The meaning of the 
equality symbol, the Boolean connectives, and the quantifiers is as expected.

A first-order formula~$\phi(x)$ with one free variable 
naturally defines the set~$\phi(G)=\{v\in V(G) : 
G\models\phi(v)\}$. We say that a formula~$\phi$ 
defines an~$f$-approximation to the dominating 
set problem on a class~$\CCC$ of graphs, if~$\phi(G)$ 
is an~$f$-approximation of a minimum dominating set for 
every graph~$G\in\CCC$. 

Observe that first-order logic is not able to count, in 
particular, no fixed formula can determine a neighbor of 
maximum degree in Line~14 of the algorithm. 
Also note however that the only place in our analysis which refers to 
the dominator function~$dom$ explicitly is  \cref{lem:insidestars}. 
The proof of the lemma in fact shows that we do not have to 
choose a vertex of maximal residual degree, but that it suffices 
to choose a neighbour of degree greater than~$4c+2c(t-1)$ if such a 
vertex exists, or any vertex, otherwise. For every fixed class of 
bounded genus, this number is a constant. We use the 
binary predicate $<$ to make a unique choice of a dominator
in this case. 

Then we define $D$ by the following formula 
\begin{align*}
\varphi_D(x) = \neg (\exists x_1\ldots \exists x_{2c}\forall y\left(E(x,y)\rightarrow \bigvee_{1\leq i\leq 2c} E(y,x_i)\right)
\end{align*}
and $D'$ by
\begin{align*}
\psi_{D'}(x) = \exists y \Big(E(x,y)\wedge \forall z\big(\varphi_D(z)\rightarrow \neg E(y,z)\big)\wedge \xi_{\max}(x,y)\Big),
\end{align*}
where $\xi_{\max}(x,y)$ states that $x$ is the maximum (residual) degree neighbour of $y$ up to threshold $4c+2c(t-1)$. We can 
express this cumbersome formula with $4c+2c(t-1)$ quantifiers. 
Note that the formulas $\varphi_D$ and $\psi_{D'}$ are different
in spirit. While $\varphi_D$ directly describes a property of 
vertices which causes them to be included in the dominating set, 
in the formula $\psi_{D'}(x)$ we state the existence
of an element which is not yet dominated by $D$ and 
which elects $x$ as a dominator. 



\section{$(1+\epsilon)$-Approximations}\label{sec:star-approx}

In this section we show how to extend techniques developed by
Czygrinow et al.~\cite{fast-planar} to find $(1+\epsilon)$-approximate
dominating set for planar graphs to graphs of 
sub-logarithmic expansion. These graphs are very general classes 
of sparse graphs, including planar graphs and all classes that exclude a
fixed minor. We focus on the dominating set problem, however, 
the approximations for the maximum weight independent set problem
and maximum matching problem proposed by Czygrinow et al.\ can be extended in a similar way. 

\smallskip
Our notation in this section closely follows that of 
Czygrinow et al.~\cite{fast-planar}. In particular, we will 
work with vertex and edge weighted graphs, that is, 
every graph $G$ is additionally equipped with two weight functions
$\omega:V(G)\rightarrow \R^+$ and $\bar{\omega}:E(G)\rightarrow \R^+$. 
If $H\subseteq G$ is a subgraph of $G$, then we write $\omega(H)$
for $\sum_{v\in V(H)}\omega(v)$ and $\bar{\omega}(H)$ for $\sum_{e\in E(H)}\bar{\omega}(e)$. If $\{G_1,\ldots,G_n\}$ is a minor model of a 
graph $H$ in a weighted graph $G$, then $H$ naturally inherits
a weight function $\omega_H$ from $G$ as follows. If $u,v\in V(H)$
are represented by the subgraphs $G_u$ and $G_v$ in the minor model, 
then $\omega(u)=\sum_{w\in V(G_v)}\omega(w)$ and if 
$\{u,v\}\in E(H)$, then $\bar{\omega}_H(\{u,v\})=\sum_{e\in E(G), e\cap V(G_u)\neq \emptyset, e\cap V(G_v)\neq \emptyset}\bar{\omega}(e)$. 

\subsection{Clustering Algorithm}

We first generalize the partitioning algorithm provided by 
Czygrinow et al.\ to graphs with sub-logarithmic expansion. 

\begin{definition}[Pseudo-Forest]
A \emph{pseudo-forest} is a directed graph $\vec{F}$ in 
which every 
vertex has out-degree at most~$1$.
\end{definition}

For a directed graph $\vec{F}$, we write $F$ for the underlying
undirected graph of~$\vec{F}$. 
The following lemma is a straightforward generalization 
of Fact~1 of~\cite{fast-planar}.

\begin{lemma}
\label{lem:pseudoforest}
Let $G$ be a graph of arboricity $a$ with an 
edge-weight function $\bar{\omega}$. There is a
distributed procedure which in two rounds finds a 
pseudo-forest $\vec{F}$ such that $F$ is a subgraph of
$G$ with $\bar{\omega}(F)\geq \bar{\omega}(G)/(2a)$. 
\end{lemma}
\begin{proof}
We run the following algorithm. For every vertex~$v$, we choose one edge
$\{v,w\}$ of largest weight, and direct it from $v$ to $w$. If we
happen to choose an edge $\{v,w\}$ for both vertices $v$ and~$w$, we
direct it from $v$ to $w$, using the larger identifier as a tie
breaker. Hence every vertex has out-degree at most one and 
the algorithm outputs a pseudo-forest $\vec{F}$. 

Let us show that $\bar{\omega}(F)\geq \bar\omega(G)/(2a)$. 
Without loss of generality we assume that $G$ has no 
isolated vertices (we make a statement about edge weights only). 
As $G$ has arboricity at most $a$, there exists a forest cover
$\mathcal{F}$ into at most $a$ forests. So one of the forests 
$T\in \mathcal{F}$ collects weight 
$\bar{\omega}(T)\ge \bar{\omega}(G)/a$. Now associate with 
each vertex $v$ of $T$ the value $w_T(v)$ which is the 
weight of the edge connecting it to 
its parent (if it exists). Similarly, 
write $w_F(v)$ for the weight of the arc $(v,w)$ or $(w,v)$ in 
$\vec{F}$ that was chosen in the algorithm. Observe that we
may be double counting edges here, but only once. Hence we have 
$\bar\omega(F)\geq \sum_{v\in V(G)}w_F(v)/2\geq \sum_{v\in V(T)}w_F(v)/2\geq \sum_{v\in V(T)}w_T(v)/2\geq \bar\omega(G)/(2a)$. 
%
%
%
\end{proof}

It is straightforward to generalize also Lemma~2 of~\cite{fast-planar}. 

\begin{lemma}[\textsc{HeavyStar}]\label{lem:heavystar}
There is a local algorithm which takes an edge weighted 
$n$-vertex graph $G$ of arboricity at most 
$a$ as input and in $\Oof(\log^*n)$ rounds computes 
a partition of $V(G)$ into vertex disjoint stars 
$H_1,\ldots, H_x\subseteq V(G)$ such that
$H=G/H_1/\ldots H_x$ has total 
weight $\bar{\omega}_H(H)\leq (1-1/(8a))\cdot \bar{\omega}(G)$.
\end{lemma}

We refrain from presenting a proof of this lemma, as the proof is literally
a copy of the proof given in \cite{fast-planar}. Czygrinow et al.~\cite{fast-planar} use only the fact that planar graphs have arboricty $3$, while
we make the statement for graphs of arboricity $a$. Hence 
only numbers must be adapted in the proof (from $24$ in their work to 
$8a$ in our case). We refer the reader to the very accessible presentation in~\cite{fast-planar}. 

\smallskip
We come to the final clustering algorithm. We fix a function 
$f(r)\in o(\log r)$
which bounds the expansion (density of depth-$r$ minors) 
of the input graphs $G$. Recall that arboricity is within 
factor~$2$ of density of subgraphs, hence the depth-$r$ minors of $G$ 
have arboricity bounded by $2f(r)$.
By iteratively taking \depthone minors, 
we obtain minors at exponential depth, 
as stated in the next lemma. 

\begin{lemma}[Proposition 4.1, statement (4.4) of \cite{nevsetvril2012sparsity}]\label{lem:iterative-density}
Taking a depth-$1$ minor of a \mbox{depth-$1$} minor
for $r$ times gives a depth\;-\;$(3^r-1)/2$ minor of $G$. 
\end{lemma}

In particular, when iterating the star contraction routine of 
Algorithm~\ref{alg:clustering}, in iteration $i$ we are dealing with a subgraph of 
arboricity $2f((3^i-1)/2)\eqqcolon g(i)$ which is sublinear in $i$. 
Hence, we may apply  \cref{lem:heavystar} with arboricity parameter
 $i$ in iteration $i$.
Note that we do not require the arboricity as an input to 
the algorithm of  \cref{lem:heavystar}. Note furthermore, that
we have \[\lim_{i\rightarrow \infty}\big(1-1/(8g(i)\big)^i\leq \lim_{i\rightarrow \infty}e^{-i/g(i)}=0,\]
hence for every $\epsilon>0$ there is a constant $i_0$ depending 
only on $\epsilon$ and $g$ (and not on the graph size~$n$) such that $\big(1-1/(8g(i_0)\big)^{i_0}\leq \epsilon$
(we may assume that the function $g$ is monotone, as the density of depth-$r$ minors cannot be smaller than the density of depth-$r'$ minors
for $r'\leq r$). 

\begin{algorithm}[H]
\caption{Clustering}
\begin{algorithmic}[1]
\vspace{2mm}
\STATE Input: $G$ with $\nabla_r(G)\leq f(r)$, $\epsilon>0$ and $i_0$ with 
$\big(1-1/(8g(i_0)\big)^{i_0}\leq \epsilon$
\smallskip
\STATE \textbf{for} $i=1,\ldots, i_0$
\STATE \qquad Call the algorithm of  \cref{lem:heavystar} to find vertex 
disjoint stars $H_1,\ldots, H_x$ in $G$
\STATE \qquad $H\leftarrow G/H_1/\ldots /H_x$ with weights
modified accordingly
\STATE \textbf{end for}
\STATE \textbf{return} $\{C_i=V(H_i) : 1\leq i\leq x\}$.
\end{algorithmic}\label{alg:clustering}
\end{algorithm}

\begin{lemma}[Clustering]
\label{lem:cluster}
Let $c\geq 1$ be a constant and let 
$G$ be a graph with $\nabla_r(G)\leq f(r)$. If the clustering algorithm gets $G$ and $\epsilon>0$ 
as input, then it returns a set of clusters
$C_1,\ldots,C_x$ partitioning $G$, such that, each cluster has
radius at most $(3^{i_0}-1)/2$ (where $i_0$ is the number of iterations in
the algorithm). Moreover, if we contract each $C_i$ to a single
vertex to obtain a graph~$H$, then $\bar{\omega}_H(H) \le \epsilon \cdot \bar{\omega}(G)$. The algorithm uses $\Oof_{\epsilon}(\log^* n)$ communication rounds. 
\end{lemma}
In the above lemma we use the notation $\Oof_{\epsilon}$ to 
express that we 
are treating all constants depending on $\epsilon$  as constants. 
\begin{proof}
As described above, the graph $G_i$ we are dealing with in iteration $i$ has
arboricity at most $2f(3^i-1)/2)=g(i)$, which is sublinear in $i$. By
applying  \cref{lem:heavystar} to~$G_i$, we compute in $\Oof(\log^*n)$ 
rounds a partition of $V(G_i)$ into vertex disjoint stars 
$H_1,\ldots, H_x\subseteq V(G_i)$ such that
$H=G_i/H_1/\ldots H_x$ has total 
weight $\bar{\omega}_H(H)\leq (1-1/(8g(i))\cdot \bar{\omega}(G_i)$.
Note that by  \cref{lem:iterative-density}, the graph $G_i$ obtained in iteration
$i$ is a depth-$((3^i-1)/2$ minor of $G$. 
Hence, by induction, after~$i$ iterations, the edge weight of the graph $G_i$
is at most $\big(1-1/(8g(i)\big)^{i}$. As argued above, there exists~$i_0$ 
such that $\big(1-1/(8g(i_0)\big)^{i_0}\leq \epsilon$, at which time we
stop the algorithm. 

As in each round we invest at most time $\Oof(\log^* n)$, in total we invest
at most $\Oof(\log^*n\cdot i_0)=\Oof_{\epsilon}(\log^* n)$ time to compute
the clustering. 
\end{proof}

\subsection{Approximation for Minimum Dominating Set}

We are ready to prove the main theorem of this section. 

\begin{theorem}\label{thm:ds-main}
There exists a deterministic distributed algorithm which gets as input
\begin{enumerate}
\item an $n$-vertex graph $G$ of sub-logarithmic expansion, 
\item a $c$-approximation of a minimum dominating set $D$ of $G$ for
some constant $c$, and 
\item a real parameter $\epsilon>0$. 
\end{enumerate}
The algorithm runs in $\Oof_{\epsilon,c}(\log^*n)$ rounds and 
outputs a $(1+\epsilon)$-approximation of a minimum dominating 
set of $G$. 
\end{theorem}

\begin{corollary}\label{crl:approx}
Let $\mathcal{C}$ be a class of graphs of sub-logarithmic expansion.
Assume there exists an algorithm which computes 
$c$-approximations of dominating sets on graphs 
from $\mathcal{C}$ in $t$ rounds. 
Then there exists an algorithm which for every $\epsilon>0$ 
computes a $(1+\epsilon)$-approximation of a minimum 
dominating set on every $n$-vertex graph $G\in\CCC$ in
$\Oof_{\epsilon,c}(t+ \log^*n)$ rounds.
\end{corollary}

We have chosen to present this extension of 
Czygrinow et al.~\cite{fast-planar}
because it connects very well to the results we obtained in the previous section. 
In particular, \cref{crl:approx} in combination with  \cref{thm:main} gives a deterministic distributed $(1+\epsilon)$-approximation algorithm in $\Oof_{\epsilon,g}(\log^* n)$ 
rounds for dominating sets on graphs
of genus at most $g$. We can similarly combine the corollary 
the result of Amiri et al.~\cite{amiri2017distributed} to 
obtain \mbox{$(1+\epsilon)$}-approximations in $\Oof(\log n)$
rounds on graphs of sub-logarithmic expansion.

\begin{proof}[Proof of  \cref{thm:ds-main}]
Let $G$ be the given input graph and let $D$ be a dominating set of $G$
with $|D|\leq c\cdot \gamma(G)$, say $D=\{w_1,\ldots, w_k\}$ (recall that $\gamma(G)$ denotes the size of a minimum size dominating set of $G$).
 Associate
each vertex $v\in V(G)\setminus D$ with one of its dominators, say with 
the one of minimum identifier, to obtain a partition $(W_1,\ldots, W_k)$
of $G$ into clusters of radius $1$. This partition is obtained in a single 
communication round. The graph $H=G/W_1/\ldots /W_k$
is a depth-$1$ minor of $G$ with $k$ vertices and at most $\nabla_1(G)\cdot k$
edges. Define an edge weight function on $E(H)$ by assigning unit weight to each
edge. Set $\delta=\epsilon/(2c\nabla_1(G))$. Apply the algorithm
of  \cref{lem:cluster} with parameter $\delta$ to 
find a partition $(V_1,\ldots, V_l)$ of $V(H)$ such that the weight between 
different clusters is at most $\delta\cdot |E(H)|$. The algorithm 
runs in $\Oof_{\delta}(\log^* n)=\Oof_{\epsilon,c}(\log^* n)$ communication rounds. 
By uncontracting 
the partitions $V_i$ and $W_i$, we obtain a partition $(U_1,\ldots,U_l)$
of $V(G)$, where each~$U_i$ has constant radius. Find an 
optimal dominating set $S_i$ of $G[U_i]$ in each subgraph $G[U_i]$ 
and return 
the union $S=\bigcup_{1\leq i\leq l} S_i$ 
of these dominating sets. As the algorithm has already learned
the subgraphs $G[U_i]$, by the infinite computational power of 
each processor in the \textit{LOCAL} model, we can compute
such a dominating set in one round. This completes the description
of the algorithm. 

Note that instead of solving the dominating set optimally on each $G[U_i]$, 
which may be considered an abuse of the \textit{LOCAL} model by some, 
we can compute a sufficiently good approximation of an optimal 
dominating set. For this, we can use the PTAS~\cite{har2015approximation} for
dominating sets on graphs of polynomial expansion. 

Since the $U_i$ form a partition of $V(G)$, it is
clear that $S$ is a dominating set of $G$. Denote by~$S^*$ a dominating set
of cardinality $\mathrm{Opt}$. Let $S'$ be obtained from $S^*$ by adding 
for each $U_i$ all vertices $w\in U_i$ which have a neighbor in a different cluster $U_j$. 
Then $S'\cap U_i$ is a dominating set of $G[U_i]$. Furthermore, we have
\[|S'|\leq |S^*|+2\delta |E(H)|\leq \gamma(G)+2c\delta \nabla_1(G)\cdot \gamma(G)=(1+\epsilon)\cdot \gamma(G).\] 

Observe that the local solutions $S_i$ cannot be worse than the solution 
$S'\cap U_i$, hence 
\[|S|=\sum_{1\leq i\leq l}|S_i|\leq \sum_{1\leq i\leq l}|S'\cap U_i|=|S'|\leq (1+\epsilon)\cdot \gamma(G).\]
\end{proof}

\section{Conclusion}\label{sec:FO}

\noindent  This paper presented the first constant round, constant factor
local MDS approximation algorithm for locally embeddable graphs, 
a class of graphs which is more general than planar graphs. 
Moreover, we have shown how our result can also be used to derive
a $\Oof(\log^*{n})$-time distributed approximation scheme for bounded
genus graphs.

Our proofs are purely combinatorial and avoid all topological arguments. 
For the family of bounded genus graphs, topological arguments
helped to improve the obtained approximation ratio in a 
preprocessing step.
We believe that this result constitutes a major step forward
in the quest for understanding for which graph families such
local approximations exist. Besides the result itself, we believe 
that our analysis introduces several new techniques which may
be useful also for the design and analysis of local algorithms
for more general graphs, and also problems beyond MDS. 
In particular, we believe that the notion of bounded depth 
minors and not the commonly used notion of excluded minors
will be the right notions in the setting of local, distributed 
computing. 

Moreover, this paper established an interesting connection between
distributed computing and logic, by presenting a local approximation
algorithm which is first-order logic definable. This also
provides an interesting new perspective on the recently introduced
notion of stone-age distributed computing~\cite{stoneage}:
distributed algorithms making minimal assumptions on the power
of a node. Avoiding counting in the arising formulas allows for example
an implementation of the algorithm in the circuit complexity class
$\mathrm{AC}^0$, that is, an implementation by circuits of polynomial
size and constant depth. 

It remains open
whether the local constant approximation result can be generalized to
sparse graphs beyond bounded genus graphs.
Also, it will be interesting to extend our study of first-order definable approximations.

%
%
%

\bibliographystyle{abbrv}
\bibliography{references}

\end{document}